\documentclass[aps, twocolumn, superscriptaddress, 10pt, nofootinbib]{revtex4-2}

\usepackage{graphicx}%
\usepackage{multirow}%
\usepackage{amsmath,amssymb,amsfonts}%
\usepackage{mathtools}
\usepackage{amsthm}%
\usepackage{mathrsfs}%
\usepackage[title]{appendix}%
\usepackage[dvipsnames]{xcolor}%
\usepackage{soul}
\usepackage{textcomp}%
\usepackage{booktabs}%
\usepackage{algorithm}%
\usepackage{algorithmicx}%
\usepackage{algpseudocode}%
\usepackage{listings}%
\usepackage{tikz}
\usetikzlibrary{quantikz2}
\usepackage{blkarray}
\usepackage{todonotes}
\usepackage{subfigure}

\usepackage[normalem]{ulem}

\usepackage{enumitem}
\usepackage{bm}

\usepackage[italicdiff]{physics}
\usepackage[colorlinks=true, linkcolor=blue, urlcolor=blue, citecolor=blue]{hyperref}

\usepackage{color}

\theoremstyle{plain}
\newtheorem{thm}{Theorem}
\newtheorem{lem}[thm]{Lemma}
\newtheorem{cor}[thm]{Corollary}

\theoremstyle{definition}
\newtheorem{ass}[thm]{Assumption}

\begin{document}

\title{Quantum framework for parameterizing partial differential equations via diagonal block-encoding}

\author{Hiroshi Yano}
\email[]{hyano@mosk.tytlabs.co.jp}
\affiliation{Toyota Central R\&D Labs., Inc., 1-4-14 Koraku, Bunkyo-ku, Tokyo 112-0004, Japan}

\author{Yuki Sato}
\affiliation{Toyota Central R\&D Labs., Inc., 1-4-14 Koraku, Bunkyo-ku, Tokyo 112-0004, Japan}

\begin{abstract}

    We study a quantum-algorithmic framework for parameterizing partial differential equations (PDEs).
    For a broad class of problems in which the discretized parameter field admits a diagonal representation, block-encodings of diagonal matrices, or diagonal block-encodings, can be used to represent spatially varying coefficients with structured, potentially complicated profiles.
    This encoding enables efficient quantum simulation of forward PDEs and extends naturally to parameter-dependent settings.
    Such simulations are a key primitive for quantum algorithms for PDE-constrained optimization, where the goal is to identify optimal design parameters.
    We illustrate the framework numerically through forward simulation and parameter design for the two-dimensional wave equation with a Gaussian parameter profile.
\end{abstract}

\maketitle

\section{Introduction}
Quantum algorithms for simulating differential equations remain a leading candidate for achieving algorithmic speedups in scientific computing \cite{berry2014Highorder,costa2019Quantum, babbush2023Exponential}.
Partial differential equations (PDEs) can also be handled on a quantum computer through discretization \cite{childs2021Highprecision,hu2024Quantum,sato2024Hamiltonian,sato2025Quantum}.
In these approaches, key subroutines, such as Hamiltonian simulation and quantum linear-system solvers, can scale only polylogarithmically with the dimension of the discretized system for suitably structured problems and access models \cite{dalzell2025Quantum}.
This motivates careful study of end-to-end implementations under realistic constraints.

Quantum simulation of PDE-constrained optimization \cite{stein2024Exponential,holscher2025Quantum,holscher2025EndtoEnd,sato2025Explicit} may offer more practical aspects, aiming to ensure that the output of quantum computation is meaningful and easily accessible.
The PDE-constrained optimization \cite{hinze2009Optimization,delosreyes2015Numerical} seeks to minimize an objective functional over control or design parameters subject to PDE constraints.
This optimization task is computationally demanding because each optimization step may require forward and adjoint computations.
In the recent study \cite{sato2025Explicit}, we introduced a quantum algorithm for PDE-constrained optimization that integrates quantum algorithms for PDE simulation and optimization.
The essential ingredient enabling their integration was the conversion from the vanilla PDE simulation to the design-parameter-dependent PDE simulation, realized by incorporating a design parameter register and performing controlled operations conditioned on that register.
However, the treatment of design parameters has not been explored in depth. For instance, the efficient implementation of translational shifts for spatially varying coefficients remains nontrivial.

For end-to-end implementations, an explicit quantum circuit construction for design-parameter-dependent PDE simulations is required.
One of the central challenges in this construction is the efficient encoding of classically described PDE data into a quantum computer.
The cost of data loading is often decisive for whether a claimed algorithmic speedup persists under practical resource estimates.
In particular, the size of the data associated with spatially varying coefficients scales with the number of spatial grid points, which underlies the potential quantum advantage, and thus must be implemented efficiently.
Numerous approaches to quantum data loading have been proposed \cite{grover2000Synthesisa,sanders2019BlackBoxa,grover2002Creatinga,babbush2018Encoding,fomichev2024Initial,berry2025Rapid,camps2023Explicit,moosa2023Lineardepth,zylberman2025Efficient}, and general frameworks for selecting among them have also emerged \cite{alonso-linaje2025Quantum}.
Consequently, there remains flexibility in the choice of methods for encoding PDE information.

In this work, we introduce a quantum-algorithmic framework for performing parameter-dependent PDE simulations.
Our approach begins with the discretization of a general linear PDE of interest, from which it follows that data loading for spatially varying coefficients reduces to the block-encoding of a diagonal matrix, or the diagonal block-encoding.
Given that diagonal block-encodings are known to admit efficient constructions without relying on amplitude oracles in certain cases, such as when structural assumptions or approximate encodings are permitted, this perspective suggests efficient data loading in quantum PDE simulation.
Through this diagonal block-encoding, we present a systematic method for encoding PDE information and show that, if the diagonal block-encodings faithfully capture the structural properties of the coefficients, the resulting implementation is efficient in terms of the gate complexity.
Notably, this viewpoint immediately extends to the implementation of design-parameter dependent operators.
Then, we discuss how the choice of representing design parameters as quantum oracles affects the available parameterizations.
With an appropriate choice of oracle, one can realize a broad range of parameterizations, including translational shifts.

The rest of our paper is organized as follows:
Section \ref{sec:background} introduces some background of this work, including quantum algorithms for PDEs and PDE-constrained optimization and quantum data loading methods.
Section \ref{sec:proposed} introduces the proposed framework, presenting a general construction of block-encodings for the forward simulation of second- and first-order linear PDEs in time, along with strategies for their parameterization by design parameters.
In Section \ref{sec:numdemo}, we numerically illustrate the framework by explicitly constructing a block-encoding for the forward simulation and the objective function in a PDE-constrained optimization for the two-dimensional wave equation, with a Gaussian parameter profile encoded in a diagonal operator. 
Lastly, we conclude by summarizing the main results of this work and discuss possible future developments in Section \ref{sec:conclusion}.

\section{Background}\label{sec:background}

\subsection{Quantum algorithm for partial differential equations}\label{subsec:background_qalgo pde}
Many quantum algorithms have been proposed to simulate unitary dynamics, known as Hamiltonian simulation algorithms.
Asymptotic improvements of these algorithms have been crucial in reducing the cost of simulation problems that are classically intractable.
Moreover, linear non-unitary dynamics can be simulated by techniques like the linear combination of Hamiltonian simulation \cite{an2023Linear,an2025Quantuma}, Schr\"odingerization \cite{jin2023Quantum,jin2024Quantuma}, and others \cite{berry2017Quantum,low2024Quantumb}.
One of the simplest linear non-unitary dynamics is generated by the time-independent, homogeneous linear ordinary differential equation (ODE)
\begin{equation}\label{eq:simple ode}
    \frac{dw(t)}{dt} = - A w(t), \quad w(0) = w_0.
\end{equation}
In the quantum query setting, the quantum complexity for simulating this dynamics is often discussed in terms of the number of queries to a unitary that prepares the initial state $\ket{w_0}$ and a unitary that prepares the operator $A$.
Furthermore, through spatial discretization, partial differential equations (PDEs) reduce to ODEs, as in Eq. \eqref{eq:simple ode}, thereby making simulation using quantum computers possible.
The outstanding challenge lies in integrating these quantum algorithms into a real-world application while reducing overall cost.
Concretely, the two main challenges are the efficient loading of PDE information, including the initial/boundary conditions and the governing equation, into a quantum computer and the efficient extraction of the desired information from the solution state.

\subsection{Quantum data loading}\label{subsec:background_data loading}
As noted above, the overall implementation for quantum PDE simulations requires that classically described PDE information be efficiently loaded into quantum circuits.
In the absence of exploitable structure in this classical data, such encoding would incur prohibitive overhead in quantum complexity.
For example, for a generic dense unstructured matrix $A \in \mathbb{R}^{N \times N}$, the following oracle is used: $O \ket{0} \ket{i} \ket{j} \mapsto (A_{ij}\ket{0}+\sqrt{1-A_{ij}}\ket{1}) \ket{i} \ket{j}$ assuming $|A_{ij}| \leq 1, \, \forall i,j$.
To implement this oracle, standard approaches require $O(N^2)$ CNOT operations \cite{camps2022FABLE}.

On the other hand, when structured data can be assumed, many efficient methods become available.
For example, if the desired quantum state has sparse amplitudes, its preparation cost scales proportionally to the number of nonzero entries \cite{fomichev2024Initial}.
Likewise, if the state exhibits low entanglement, the matrix product state (MPS) formalism enables efficient implementation \cite{fomichev2024Initial,berry2025Rapid}.

In this paper, we focus on block-encodings for diagonal matrices, which we hereinafter refer to as diagonal block-encodings.
In particular, efficient approaches exist when the entries of a diagonal matrix are given as function values \cite{zylberman2025Efficient,motlagh2024Generalized,mcardle2025Quantum,rosenkranz2025Quantum,mori2024Efficient}.
As a notable example, diagonal block-encodings can be used to prepare the Kaiser window function on ancilla qubits in quantum phase estimation.
In the context of encoding parameter information in PDEs, these approaches have already been exploited as well.
For example, Ref.~\cite{li2023Efficient} employs a circuit similar to that of Ref.~\cite{mcardle2025Quantum} to construct a diagonal block-encoding of an elliptic operator with periodic boundary conditions.

\subsection{Quantum PDE-constrained optimization}
Solving optimization problems subject to constraints in the form of PDEs is one of the most challenging computational problems since it requires the interaction of optimization techniques and numerical simulation of PDEs.
For PDEs of the form of Eq. \eqref{eq:simple ode} after spatial discretization and with the introduction of design parameters $\xi \in \Xi$, where $\Xi$ denotes the design space (including both continuous and discrete variables), PDE-constrained optimization problems over $\xi$ can be formulated as 
\begin{align}
    &\min_{\xi \in \Xi} \mathcal{F}(w(t;\xi)) \nonumber\\
    &\mathrm{s.t.} \quad 
    \frac{dw(t;\xi)}{dt} = - A(\xi) w(t;\xi), \quad w(0;\xi) = w_0,
\end{align}
where $\mathcal{F}(w(t;\xi))$ denotes the objective function for the parameterized PDE solution state $w(t;\xi)$.
The expectation is that, by leveraging the speedups of quantum PDE simulations, this problem can also be solved efficiently on a quantum computer.
In particular, the key technique for linking the optimization with the PDE simulation can be formulated as the construction of a block-encoding of the objective function $\mathcal{F}(\cdot)$, $\hat{\mathcal{F}} \coloneqq \sum_\xi \mathcal{F}(w(t;\xi)) \ketbra{\xi}{\xi}$, which enables the use of quantum optimization algorithms \cite{gilyen2019Optimizing,catli2025Exponentially}.
Although employing quantum PDE simulation as a subroutine requires greater quantum resources than the standalone simulation, it is expected to address the output problem, i.e., enabling the extraction of the desired information, namely the optimal design-parameter value, from the PDE solution state.
For constructing the block-encoding of $\hat{\mathcal{F}}$, it is required to simulate a design-parameter-dependent ODE
\begin{equation}\label{eq:design parameter ode}
    \frac{d\tilde{w}(t;\xi)}{dt} = - \tilde{A}(\xi) \tilde{w}(t;\xi), \quad
    \tilde{A}(\xi) \coloneqq \sum_\xi \ketbra{\xi}{\xi} \otimes A(\xi),
\end{equation}
with an initial state $\tilde{w}(0;\xi) = \ket{\psi} \otimes w_0$, where $\ket{\psi} = \sum_\xi \psi_\xi \ket{\xi}$ is an arbitrary quantum state with $\psi_\xi \in \mathbb{C}$.
Since the above equation can be viewed as the linear ODE in Eq. \eqref{eq:simple ode}, it can be simulated by the quantum algorithms introduced in Section \ref{subsec:background_qalgo pde}.
The central challenge is to reduce the cost of each query to $\tilde{A}(\xi)$ while retaining flexibility in the parameterization of $A(\xi)$; addressing this challenge is the aim of this work.

\section{Proposed framework}\label{sec:proposed}
We present a quantum-algorithmic framework for simulating PDEs with spatially varying coefficients and its extension to design-parameter-dependent settings.
Section \ref{subsec:proposed_problem} summarizes the class of PDEs we consider and derives the associated operator $A$ appearing in Eq. \eqref{eq:simple ode}.
In Section \ref{subsec:proposed_block encoding}, we show how to encode the operator $A$ into a quantum circuit via block-encoding constructions.
Finally, Section \ref{subsec:proposed_param} introduces strategies for parameterizing $A$ by design parameters $\xi$, enabling the simulation of the design-parameter-dependent ODE in Eq. \eqref{eq:design parameter ode}.

\subsection{Problem of interest}\label{subsec:proposed_problem}
We consider the following second- and first-order linear PDEs in time, defined on an open bounded set $\Omega \subset \mathbb{R}^d$ with the spatial coordinate $\bm{x}$, where $d$ denotes the spatial dimension:
\begin{align}\label{eq:second_order_linear_pde}
    \varrho(\bm{x}) \frac{\partial^2 u(t,\bm{x})}{\partial t^2} + \zeta(\bm{x}) \frac{\partial u(t,\bm{x})}{\partial t} - \nabla \cdot \kappa(\bm{x}) \nabla u(t,\bm{x}) \nonumber\\ + \gamma(\bm{x}) u(t,\bm{x}) = 0,
\end{align}
and
\begin{align}\label{eq:first_order_linear_pde}
    \frac{\partial u(t,\bm{x})}{\partial t} - \nabla \cdot \kappa(\bm{x}) \nabla u(t,\bm{x}) + \bm{\beta}(\bm{x}) \cdot \nabla u(t,\bm{x}) \nonumber\\ + \gamma(\bm{x}) u(t,\bm{x}) = 0,
\end{align}
under appropriate initial condition.
Among various mappings for converting these PDEs into the form of the ODE \eqref{eq:simple ode}, we follow the same procedure of Ref. \cite{sato2025Quantum}, in which Eqs. \eqref{eq:second_order_linear_pde} and \eqref{eq:first_order_linear_pde} are represented as
\begin{align}
    &\frac{\partial}{\partial t} 
    \begin{pmatrix}
        \sqrt{\varrho} \dot{u} \\
        \sqrt{\kappa} \nabla u \\
        \sqrt{\gamma} u
    \end{pmatrix} \nonumber\\
    =& -
    \begin{pmatrix}
        \frac{\zeta(\bm{x})}{\varrho(\bm{x})} & - \frac{1}{\sqrt{\varrho(\bm{x})}} \nabla^\top \sqrt{\kappa(\bm{x})} & \sqrt{\frac{\gamma(\bm{x})}{\varrho(\bm{x})}} \\
        - \sqrt{\kappa(\bm{x})} \nabla \frac{1}{\sqrt{\varrho(\bm{x})}} & 0 & 0 \\
        - \sqrt{\frac{\gamma(\bm{x})}{\varrho(\bm{x})}} & 0 & 0
    \end{pmatrix} \nonumber\\
    &\times
    \begin{pmatrix}
        \sqrt{\varrho} \dot{u} \\
        \sqrt{\kappa} \nabla u \\
        \sqrt{\gamma} u
    \end{pmatrix},
    \label{eq:second_order_linear_pde_mod}
\end{align}
and
\begin{equation}
    \frac{\partial u}{\partial t} = \left( \nabla \cdot \kappa(\bm{x}) \nabla - \bm{\beta}(\bm{x}) \cdot \nabla - \gamma(\bm{x}) \right) u.
    \label{eq:first_order_linear_pde_mod}
\end{equation}
We discretize the domain $\Omega$ into a grid of $N = \prod_{\mu=0}^{d-1} N_\mu$ points (assuming $N_\mu = 2^{n_\mu}$ in this paper) and represent the scalar field $u(x)$ as $u \coloneqq [u(x^{[0]}), u(x^{[1]}), ..., u(x^{[N-1]})]^\top$, where $x^{[j]}$ denotes the coordinate of the $j$-th grid point.
Then, through the introduction of a diagonal operator
\begin{equation}\label{eq:diagonal op}
    C_f \coloneqq \sum_{j=0}^{2^n-1} f(\bm{x}^{[j]}) \ketbra{j}{j},
\end{equation}
representing the spatially varying coefficient $f(\bm{x})$, together with the forward and backward difference operators $D_\mu^+$ and $D_\mu^-$, the operator $A$ corresponding to Eqs. \eqref{eq:second_order_linear_pde_mod} and \eqref{eq:first_order_linear_pde_mod} can be written as 
\begin{align}
    A^{\mathrm{(2nd)}} &= \ketbra{0}{0} \otimes C_\varrho^{-1} C_\zeta - \sum_{\mu=0}^{d-1} \ketbra{0}{\mu+1} \otimes C_\varrho^{-\frac{1}{2}} D_\mu^+ C_\kappa^{\frac{1}{2}}\nonumber\\
    &+ \ketbra{0}{d+1} \otimes C_\varrho^{-\frac{1}{2}} C_\gamma^{\frac{1}{2}} - \sum_{\mu=0}^{d-1} \ketbra{\mu+1}{0} \otimes C_\kappa^{\frac{1}{2}} D_\mu^- C_\varrho^{-\frac{1}{2}} \nonumber\\
    &- \ketbra{d+1}{0} \otimes C_\varrho^{-\frac{1}{2}} C_\gamma^{\frac{1}{2}} + \mathrm{b.c.}, 
\end{align}
and
\begin{align}
    A^{\mathrm{(1st)}} &= - \left( \frac{1}{2} \sum_{\mu=0}^{d-1} (D_\mu^+ C_\kappa D_\mu^- + D_\mu^- C_\kappa D_\mu^+) \right. \nonumber\\
    &\left. + \sum_{\mu=0}^{d-1} (C_{\beta_\mu^+} D_\mu^- + C_{\beta_\mu^-} D_\mu^+) + C_\gamma \right) + \mathrm{b.c.},
\end{align}
respectively.
Here, $\beta_\mu^+(\bm{x}) = \max(\beta_\mu(\bm{x}),0)$ and $\beta_\mu^-(\bm{x}) = \min(\beta_\mu(\bm{x}),0)$, while $\mathrm{b.c.}$ denotes the remaining terms arising from boundary conditions.
We refer to Ref. \cite{sato2025Quantum} for the details of those transformations.
We note that other mappings, such as those proposed in Refs. \cite{schade2024Quantum,guseynov2025Gate}, can also yield a similar operator structure, in the sense that it consists of a linear combination of products of a diagonal operator $C_f$ and a difference operator $D$. %

In light of the above, our first goal is to construct a block-encoding of the operator $A$, focusing on a diagonal block-encoding of $C_f$.
Our second goal is to outline approaches for constructing a block-encoding of the operator $\tilde{A}(\xi) = \sum_\xi \ketbra{\xi}{\xi} \otimes A(\xi)$ for parameterized PDE simulations.

\subsection{Block-encoding}\label{subsec:proposed_block encoding}

Here, we describe how to construct a block-encoding of the operator $A$ with spatially varying coefficients encoded through diagonal block-encodings, and analyze its query and gate complexity.
There are some approaches \cite{grover2002Creatinga,mcardle2025Quantum,motlagh2024Generalized,zylberman2025Efficient,rosenkranz2025Quantum,mori2024Efficient} for constructing a diagonal block-encoding of $C_f$.
Since the appropriate construction depends on the specific form of $f$, we do not restrict to a particular method here; instead, we assume access to an $(\alpha_f, a_f, \epsilon_f)$-block-encoding of $C_f$ and focus on how to build a block-encoding of the operator $A$ from it.
Our analysis first considers the query complexity in the presence of general $(\alpha_f, a_f, \epsilon_f)$-block-encoding of $C_f$. 
Next, we derive the gate complexity under the assumption that the diagonal block-encoding is implemented using the LCU Fourier method of Ref. \cite{rosenkranz2025Quantum}. 
By analogy with this analysis, one can derive the gate complexity for any alternative diagonal block-encoding by substituting its corresponding gate cost.

\subsubsection{Second-order linear PDE}
We first assume access to block-encodings of $C_\varrho^{-\frac{1}{2}}, C_\kappa^{\frac{1}{2}}, C_\zeta$, and $C_\gamma^{\frac{1}{2}}$.
\begin{ass}[Block-encodings for coefficients, $A^{\mathrm{(2nd)}}$]
    \label{ass:BE_coeff_A2nd}
    We have access to the following block-encodings:
    \begin{itemize}
        \item $U_{\frac{1}{\sqrt{\varrho}}}$, an $(\alpha_{\frac{1}{\sqrt{\varrho}}},a_{\frac{1}{\sqrt{\varrho}}},\epsilon_{\frac{1}{\sqrt{\varrho}}})$-block-encoding of $C_\varrho^{-\frac{1}{2}}$, and its inverse
        \item $U_{\sqrt{\kappa}}$, an $(\alpha_{\sqrt{\kappa}},a_{\sqrt{\kappa}},\epsilon_{\sqrt{\kappa}})$-block-encoding of $C_\kappa^{\frac{1}{2}}$
        \item $U_\zeta$, an $(\alpha_\zeta,a_\zeta,\epsilon_\zeta)$-block-encoding of $C_\zeta$
        \item $U_{\sqrt{\gamma}}$, an $(\alpha_{\sqrt{\gamma}},a_{\sqrt{\gamma}},\epsilon_{\sqrt{\gamma}})$-block-encoding of $C_\gamma^{\frac{1}{2}}$
    \end{itemize}
\end{ass}

Note that the residual terms for boundary conditions can be included in the difference operators, $D_\mu^+$ and $D_\mu^-$, and we denote the modified version of the difference operator as $D_\mu^{+'}$ and $D_\mu^{-'}$.
For simplicity, we assume that the operator norm of both operators for every $\mu$ is bounded by $\alpha_D$, and at most $a_D$ ancilla qubits are used for their block-encodings.
\begin{ass}[Block-encodings for difference operators]
    \label{ass:BE_diff}
    We have access to the following block-encodings:
    \begin{itemize}
        \item $U_{D_\mu^{+'}}$, an $(\alpha_D,a_D,0)$-block-encoding of $D_\mu^{+'}$
        \item $U_{D_\mu^{-'}}$, an $(\alpha_D,a_D,0)$-block-encoding of $D_\mu^{-'}$
    \end{itemize}
    for $\mu=0,...,d-1$.
\end{ass}
For example, see Ref. \cite[Appendix A]{sato2025Explicit} for the explicit construction.

Given the assumptions, we show the query complexity of the block-encoding of $A^{\mathrm{(2nd)}}$.
\begin{thm}[Query complexity for a block-encoding of $A^{\mathrm{(2nd)}}$]\label{thm:query_A2nd}
    Suppose we have access to block-encodings in Assumption \ref{ass:BE_coeff_A2nd} and \ref{ass:BE_diff}.
    Then, we can construct a $((d+2)\max(\alpha_{\frac{1}{\sqrt{\varrho}}}^2 \alpha_\zeta, \alpha_{\sqrt{\kappa}} \alpha_D \alpha_{\frac{1}{\sqrt{\varrho}}}, \alpha_{\frac{1}{\sqrt{\varrho}}} \alpha_{\sqrt{\gamma}}),a_{\frac{1}{\sqrt{\varrho}}}+\max(a_{\sqrt{\kappa}},a_\zeta,a_{\sqrt{\gamma}})+a_D+1, (d+2)\max(\alpha_D \alpha_{\sqrt{\kappa}} \epsilon_{\frac{1}{\sqrt{\varrho}}} + \alpha_{\frac{1}{\sqrt{\varrho}}} \alpha_D \epsilon_{\sqrt{\kappa}}, \alpha_{\frac{1}{\sqrt{\varrho}}} \epsilon_{\sqrt{\gamma}} + \alpha_{\sqrt{\gamma}} \epsilon_{\frac{1}{\sqrt{\varrho}}}, \alpha_{\frac{1}{\sqrt{\varrho}}}^2 \epsilon_\zeta + 8 \alpha_\zeta \sqrt{\epsilon_{\frac{1}{\sqrt{\varrho}}}}))$-block-encoding of $A^{\mathrm{(2nd)}}$ with the controlled version of $2d+2$ uses of $U_{\frac{1}{\sqrt{\varrho}}}$, $1$ use of $U_{\frac{1}{\sqrt{\varrho}}}^\dagger$, $2d$ uses of $U_{\sqrt{\kappa}}$, $1$ use of $U_\zeta$, $1$ use of $U_{\sqrt{\gamma}}$, $d$ uses of $U_{D_\mu^{+'}}$, and $d$ uses of $U_{D_\mu^{-'}}$.
\end{thm}

Next, we analyze the gate complexity required to construct the above block-encoding.
To this end, we assume that all spatially varying coefficients can be approximated by a $K$-th order Fourier series, enabling an explicit construction of the block-encodings in Assumption \ref{ass:BE_coeff_A2nd}.
\begin{thm}[Gate complexity for a block-encoding of $A^{\mathrm{(2nd)}}$]\label{thm:gate_A2nd}
    Let $n$ be divisible by $d$ and $n_\mu = n/d$ for all $\mu$.
    When we use LCU Fourier \cite{rosenkranz2025Quantum} for Assumption \ref{ass:BE_coeff_A2nd} with maximum degree $K$ and the method of \cite[Appendix A]{sato2025Explicit} for Assumption \ref{ass:BE_diff}, constructing the block-encoding of $A^{\mathrm{(2nd)}}$ in Theorem \ref{thm:query_A2nd} needs $\mathcal{O}(d K^d + d n \log K + n^2)$ two-qubit gates.
\end{thm}
Proof details for the above theorems are provided in Appendix \ref{subsec:app_A2nd}.
We can see that, under the assumption that the spatially varying coefficients can be efficiently approximated by Fourier expansions, following the construction of the block-encoding of $A^\mathrm{(2nd)}$ derived in Theorem \ref{thm:query_A2nd} allows one to avoid the complexity proportional to the number of grid points, i.e., the exponential scaling of $\mathcal{O}(2^n)$. 
We emphasize that this observation is not restricted to Fourier-based construction, but applies more generally to any method that implements the block-encoding of $C_f$ with gate complexity polynomial in $n$.
A crucial issue, however, is the treatment of the resulting block-encoding error of $A^\mathrm{(2nd)}$. Since the argument proceeds analogously, we omit it here and address it in the discussion of $A^\mathrm{(1st)}$.

\subsubsection{First-order linear PDE}
For simplicity, we assume that the operator norm of both $C_{\beta_\mu^+}$ and $C_{\beta_\mu^-}$ for every $\mu$ is bounded by $\alpha_\beta$, and at most $a_\beta$ ancilla qubits are used for their block-encodings.
\begin{ass}[Block-encodings for coefficients, $A^{\mathrm{(1st)}}$]
    \label{ass:BE_coeff_A1st}
    We have access to the following block-encodings:
    \begin{itemize}
        \item $U_\kappa$, an $(\alpha_\kappa,a_\kappa,\epsilon_\kappa)$-block-encoding of $C_\kappa$
        \item $U_\gamma$, an $(\alpha_\gamma,a_\gamma,\epsilon_\gamma)$-block-encoding of $C_\gamma$
        \item $U_{\beta_\mu^+}$, an $(\alpha_\beta,a_\beta,\epsilon_\beta)$-block-encoding of $C_{\beta_\mu^+}$
        \item $U_{\beta_\mu^-}$, an $(\alpha_\beta,a_\beta,\epsilon_\beta)$-block-encoding of $C_{\beta_\mu^-}$
    \end{itemize}
\end{ass}

Given the above assumption as well as that for difference operators, we obtain the following query complexity.
\begin{thm}[Query complexity for a block-encoding of $A^{\mathrm{(1st)}}$]\label{thm:query_A1st}
    Suppose we have access to block-encodings in Assumption \ref{ass:BE_coeff_A1st} and \ref{ass:BE_diff}.
    Then, we can construct a $((3d+1)\max(\alpha_\kappa \alpha_D^2, \alpha_\beta \alpha_D, \alpha_\gamma), \max(a_\kappa,a_\beta)+2a_D+\lceil\log_2(4d+1)\rceil, (3d+1)\max(\alpha_D^2 \epsilon_\kappa, \alpha_D \epsilon_\beta, \epsilon_\gamma))$-block-encoding of $A^{\mathrm{(1st)}}$ with the controlled version of $2d$ uses of $U_\kappa$, $d$ uses of $U_{\beta_\mu^+}$, $d$ uses of $U_{\beta_\mu^-}$, $1$ use of $U_\gamma$, $3d$ uses of $U_{D_\mu^{+'}}$, and $3d$ uses of $U_{D_\mu^{-'}}$.
\end{thm}

Again, the availability of explicit block-encoding construction in Assumptions \ref{ass:BE_coeff_A1st} and \ref{ass:BE_diff} enables us to analyze the gate complexity.
\begin{thm}[Gate complexity for a block-encoding of $A^{\mathrm{(1st)}}$]\label{thm:gate_A1st}
    Let $n$ be divisible by $d$ and $n_\mu = n/d$ for all $\mu$.
    When we use LCU Fourier \cite{rosenkranz2025Quantum} for Assumption \ref{ass:BE_coeff_A1st} with maximum degree $K$ and the method of \cite[Appendix A]{sato2025Explicit} for Assumption \ref{ass:BE_diff}, constructing the block-encoding of $A^{\mathrm{(1st)}}$ in Theorem \ref{thm:query_A1st} needs $\mathcal{O}(d K^d + d n \log K + n^2)$ two-qubit gates.
\end{thm}

Proof details for the above theorems are provided in Appendix \ref{subsec:app_A1st}.
Again, assuming efficient Fourier approximations of the coefficients, the method of Theorem \ref{thm:query_A1st} circumvents exponential scaling in $n$.
Furthermore, a more detailed analysis can be made by determining the degree $K$ of approximation required to achieve a desired final block-encoding error for $A^\mathrm{(1st)}$. 
In Corollary \ref{cor:gate_A1st_example} of the Appendix, it is shown that, when spatially varying coefficients can be well approximated by Fourier series, a block-encoding of $A^\mathrm{(1st)}$ with error no greater than $\epsilon$ can be constructed with $\mathrm{poly}(n)$ gate complexity.

\subsection{Parameterization}\label{subsec:proposed_param}
Here we consider introducing design parameters $\xi$ for parameterizing spatially varying coefficients.
In terms of the operator $A^{\mathrm{(2nd)}}$ and $A^{\mathrm{(1st)}}$, the diagonal operators $C_f$ defined in Eq. \eqref{eq:diagonal op} admit a parameterized generalization, denoted by $C_f(\xi) \coloneqq \sum_{j=1}^{2^n-1} f(\bm{x}^{[j]};\xi) \ketbra{j}{j}$.
In the quantum PDE-constrained optimization, the block-encoding of $\tilde{A}(\xi)$ is required to simulate the design-parameter-dependent ODE in Eq. \eqref{eq:design parameter ode}.
Notably, since the same procedure as presented in the previous subsection can be followed, the construction of a block-encoding of $\tilde{A}(\xi)$ reduces to that of a block-encoding of
\begin{equation}\label{eq:diagonal operator}
    \tilde{C}_f(\xi) \coloneqq \sum_\xi \ketbra{\xi}{\xi} \otimes C_f(\xi)
    = \sum_\xi \sum_j f(\bm{x}^{[j]};\xi) \ketbra{\xi}{\xi} \otimes \ketbra{j}{j}.
\end{equation}
Accordingly, the primary focus here is the construction of a block-encoding of $\tilde{C}_f(\xi)$.
To this end, we focus in particular on which quantum properties are used, or equivalently, which quantum oracles are employed to introduce the design parameters $\xi$.
In this paper, we discuss three quantum oracles for a function $g(\xi)$ of the variable $\xi$, as formulated in Ref. \cite{gilyen2019Optimizing,gilyen2019Quantuma}.
Given that our target is a diagonal block-encoding, organizing the discussion around a diagonal block-encoding from the outset simplifies the presentation.

The first is a probability oracle,
\begin{align}\label{eq:probability oracle}
    U_\mathrm{prob}&: \ket{0^a}\ket{\xi} \mapsto \nonumber\\
    & (\sqrt{p(\xi)}\ket{0}\ket*{\psi_\xi^{(0)}}+\sqrt{1-p(\xi)}\ket{1}\ket*{\psi_\xi^{(1)}})\ket{\xi},
\end{align}
where $p(\cdot)$ is a real-valued function taking values in $[0,1]$, and $\ket*{\psi_\xi^{(0)}}$ and $\ket*{\psi_\xi^{(1)}}$ are arbitrary normalized quantum states.
Note that $W \coloneqq U_\mathrm{prob}^\dagger (Z \otimes I^{\otimes a-1}) U_\mathrm{prob}$ can be recognized as a $(1,a,0)$-diagonal-block-encoding of 
\begin{equation}\label{eq:block encoding style of probability oracle}
    \sum_\xi g(\xi) \ketbra{\xi}{\xi},
\end{equation}
with $g(\xi) = 2p(\xi)-1$ \cite{gilyen2019Quantuma}.
The second is a phase oracle,
\begin{equation}\label{eq:phase oracle}
    U_\mathrm{phase}: \ket{0^a}\ket{\xi} \mapsto e^{i g(\xi)} \ket{0^a}\ket{\xi},
\end{equation}
for a real-valued function $g(\cdot)$ that takes values in $[-1,1]$.
The phase is not commonly used in loading data, but there may be cases where phase information is useful in efficiently representing the spatially varying coefficients.
Note that the phase oracle can be recognized as a $(1,a,0)$-diagonal-block-encoding of
\begin{equation}\label{eq:block encoding style of phase oracle}
    \sum_\xi e^{i g(\xi)} \ketbra{\xi}{\xi}.
\end{equation}
The third is a binary oracle,
\begin{equation}\label{eq:binary oracle}
    U_\mathrm{binary}: \ket{0^a}\ket{\xi} \mapsto \ket{\tilde{g}(\xi)}\ket{\xi},
\end{equation}
where $g(\cdot)$ is a real-valued function and $\tilde{g}(\cdot)$ denotes its fixed-point binary representation.
For example, it may be used to parameterize the spatial variable $x$ by $\xi$.
Since the computational basis is typically used to represent spatial coordinates in quantum PDE simulations, the corresponding bit strings can be used to control spatial transformations.
The binary oracle consists solely of bitwise operations and therefore does not admit a corresponding block-encoding representation.

In the remainder of this subsection, we illustrate several useful parameterizations by presenting explicit constructions of block-encodings for $\tilde{C}_f(\xi)$.
For concreteness, we build upon the result of Ref. \cite{rosenkranz2025Quantum} for the diagonal block-encoding in Eqs. \eqref{eq:block encoding style of probability oracle} and \eqref{eq:block encoding style of phase oracle}, as well as Ref. \cite{sanders2019BlackBoxa} for the binary oracle.

\subsubsection{Probability}
We describe the most straightforward use case of the probability oracle.
For the function parameterization of the form $f(x;\xi) = \sum_k \xi_k \phi_k(x)$ for some basis function $\{\phi_k(x)\}$,
we decompose $\tilde{C}_f(\xi)$ as 
\begin{align}
    \tilde{C}_f(\xi) &= \sum_{k=0}^{M-1} ( I^{\otimes m k} \otimes \sum_{\xi_k} \xi_k \ketbra{\xi_k}{\xi_k} \otimes I^{\otimes m(M-1-k)} \otimes C_{\phi_k} ),
\end{align}
with a discretized operator $C_{\phi_k}$ for $\phi_k(x)$.
This block-encoding can be constructed by the linear combination of block-encoded matrices using the diagonal-block-encoding of the form \eqref{eq:block encoding style of probability oracle} with $g(\xi) = \xi$, as well as the access to the block-encodings of $\{C_{\phi_k}\}$, as similarly considered in Ref. \cite{sato2025Explicit}.
At the circuit level, parameterizing the block-encoding of $C_f$ reduces to adding the diagonal block-encoding of $\sum_\xi \xi \ketbra{\xi}{\xi}$.

This representation covers many spatially varying coefficients.
When the bases represent local region, we can linearly adjust an amplitude at the region.
When the function admits an efficient representation of the Chebyshev approximation, the Chebyshev coefficients can be parameterized.
For example, a parameterized function $f((x,y);\xi) = -(x^2+y^2)+\xi$ can be used to adjust the radius of a circle.
Its block-encoding can be constructed by linearly combining block-encodings for $I^{\otimes m} \otimes \sum_x x^2 \ketbra{x}{x} \otimes I^{\otimes n_y}$, $I^{\otimes m} \otimes I^{\otimes n_x} \otimes \sum_y y^2 \ketbra{y}{y}$, and $\sum_\xi \xi \ketbra{\xi}{\xi} \otimes I^{\otimes n_x + n_y}$. 
For the diagonal operator like $\sum_x x^2 \ketbra{x}{x}$, one may utilize the $(1,\lceil\log_2(n_x)\rceil,0)$-block-encoding of the $k$-th Chebyshev basis function proposed in Ref. \cite{rosenkranz2025Quantum}.
Then, one can readily obtain a diagonal block-encoding for a parameterized function $f((x,y);\xi) = \mathrm{sign}(-(x^2+y^2)+\xi)$ by further employing QSVT for a polynomial approximation of $\mathrm{sign}(\cdot)$ with a diagonal block-encoding for $-(x^2+y^2)+\xi$. 
Such a function can approximate a spatially varying coefficient that takes the value $1$ inside the circle and the value $-1$ outside the circle, with the circle’s radius as a parameter $\xi$.

\subsubsection{Phase}
When the function admits an efficient representation with a phase as a parameter, the phase oracle may be useful.
The primary representation is the Fourier series approximation.
For example, when the function can be approximated as $f(x;\xi) = \sum_k c_k e^{i \pi k (x+\xi)}$, the phase oracle, or the diagonal block-encoding of the form \eqref{eq:block encoding style of phase oracle} with $g(\xi) = \pi k \xi$, can be used to represent a spatial shift.
One may utilize the $(1,0,0)$-block-encoding of the $k$-th Fourier basis function proposed in Ref. \cite{rosenkranz2025Quantum}.

In the numerical demonstration for the two-dimensional system in Section \ref{sec:numdemo}, given that the function $f(x,y)$ can be represented by a finite Fourier series of the form
\begin{equation}
    f(x,y) \approx f_{(K_x,K_y)}^F(x,y) = \sum_{k=-K_x}^{K_x} \sum_{l=-K_y}^{K_y} c_{k,l} e^{i \pi k x} e^{i \pi l y},
\end{equation}
we deal with the spatial shift $\xi = (\xi_x,\xi_y)$ as a separate phase factor:
\begin{align}
    &f_{(K_x,K_y)}^F(x+\xi_x,y+\xi_y) = \nonumber\\
    &\quad \sum_{k=-K_x}^{K_x} \sum_{l=-K_y}^{K_y} c_{k,l} e^{i \pi k \xi_x} e^{i \pi l \xi_y} e^{i \pi k x} e^{i \pi l y}.
\end{align}
Following Ref. \cite{rosenkranz2025Quantum}, through the discretization of variables $x$, $y$, $\xi_x$, and $\xi_y$, $\tilde{C}_{f}(\xi)$ can be approximated as 
\begin{align}
    &\tilde{C}_{f_{(K_x,K_y)}^F}(\xi) \\
    =& \sum_{\xi_x,\xi_y} \ketbra{\xi_y,\xi_x}{\xi_y,\xi_x} \otimes C_f(\xi_x,\xi_y) \nonumber\\
    =& \left( U_{\xi_x}^{-K_x} \otimes U_{\xi_y}^{-K_y} \otimes U_{x}^{-K_x} \otimes U_{y}^{-K_y} \right) \nonumber\\
    &\times \left( \sum_{\mu_x=0}^{2K_x} \sum_{\mu_y=0}^{2K_y} \lambda_{\mu_x,\mu_y} U_{\xi_x}^{\mu_x} \otimes U_{\xi_y}^{\mu_y} \otimes U_{x}^{\mu_x} \otimes U_{y}^{\mu_y} \right),
    \label{eq:diagonal operator spatial shift}
\end{align}
where we introduce the change of variables, $\mu_x=k_x+K_x$, $\mu_y=k_y+K_y$, $\lambda_{\mu_x,\mu_y}=c_{\mu_x-K_x,\mu_y-K_y}$, and an $n$-qubit unitary $U_r = e^{i \pi H_r^F}$ with $H_r^F = \sum_{j=0}^{2^n-1} j/(2^n-1) \ketbra{j}{j}$ on the register $r = x, y, \xi_x, \xi_y$.
Since the representation in Eq. \eqref{eq:diagonal operator spatial shift} is the form of a linear combination of unitaries (LCU), a $(\sum_{\mu_x=0}^{2K_x} \sum_{\mu_y=0}^{2K_y} |\lambda_{\mu_x,\mu_y}|, \lceil\log_2(2K_x+2K_y+2)\rceil, 0)$-block-encoding of $\tilde{C}_{f_{(K_x,K_y)}^F}(\xi)$ can be constructed via the LCU method.
From this example, we see that just adding controlled phase shift gates on the design parameter register can yield the construction of the block-encoding of the diagonal operator in Eq. \eqref{eq:diagonal operator}. 
In the numerical demonstration, we explicitly construct a diagonal block-encoding for a truncated Fourier series approximation of two-dimensional Gaussian function parameterized by a central shift as $(\xi_x, \xi_y)$.

\subsubsection{Binary}
By definition, the binary oracle is not provided in block-encoding form; nevertheless, it can be crucially useful in some contexts.
In particular, when constructing a block-encoding of a function, it is advantageous for handling discontinuities, as suggested in Ref. \cite{mcardle2025Quantum}.
Consider a piecewise representation of the function $f(x;\xi) = \sum_k c_k \phi_k(x;\xi)$, where $\phi_k(x;\xi) = \phi_k(x)$ on the interval $\xi_{k-1} \leq x < \xi_k$.
In such a case, given the operation
\begin{equation}
    \ket{x} \ket{\xi_k} \ket{0}_\mathrm{flag} \mapsto
    \begin{cases}
    \ket{x} \ket{\xi_k} \ket{0}_\mathrm{flag} & \text{if $x < \xi_k$}, \\
    \ket{x} \ket{\xi_k} \ket{1}_\mathrm{flag} & \text{if $x \geq \xi_k$},
    \end{cases}
\end{equation}
introduced in Ref. \cite{sanders2019BlackBoxa}, one can construct a diagonal block-encoding of the piecewise function with an adjustable interval parameterized by $\xi_k$, using the controlled version of the diagonal block-encoding for $\phi_k(x)$, with the flag qubits serving as the control qubits.
By effectively using the binary oracle, it may be possible to efficiently implement complex functions that cannot be represented by a single global polynomial.

\section{Numerical demonstration}\label{sec:numdemo}
Here we demonstrate the proposed framework numerically with (i) forward simulation and (ii) parameter design for two-dimensional acoustic simulation using PennyLane \cite{bergholm2022PennyLane}.
We also use pyqsp \cite{chao2020Finding,dong2021Efficient,martyn2021Grand} to determine the phase angles required for QSVT-based Hamiltonian simulation.
The governing wave equation is given as 
\begin{equation}
    \frac{1}{c(x,y)^2} \frac{\partial^2 u(t,x,y)}{\partial t^2} = \nabla^2 u(t,x,y), \quad (x,y) \in (0,1)^2, 
\end{equation}
where $c(x,y) = \sqrt{1/\varrho(x,y)}$ represents the speed of sound that can spatially vary and we assume it follows a two-dimensional shifted Gaussian profile
\begin{align}
    c(x,y) = - \exp(- \left(\frac{(x-1/2)^2}{2\cdot(1/20)^2} + \frac{(y-1/2)^2}{2\cdot(1/5)^2} \right)) + 1, &\nonumber\\
    \quad (x,y) \in (0,1)^2. &
\end{align}
Discretizing the domain $(0, 1)^2$ into a uniform grid as discussed in Sec.~\ref{subsec:proposed_problem}, the governing equation can be written in the form of the ODE \eqref{eq:simple ode}, where $w(t) = (C^{-1} \frac{\partial u}{\partial t}, \nabla_x u, \nabla_y u, 0)^T$ and 
\begin{equation}\label{eq:acoustic A}
    A 
    = - \begin{pmatrix}
        0 & C D_x^+ & C D_y^+ & 0 \\
        D_x^- C & 0 & 0 & 0 \\
        D_y^- C & 0 & 0 & 0 \\
        0 & 0 & 0 & 0 \\
      \end{pmatrix},
\end{equation}
with $C = \sum_{x,y} c(x,y) \ketbra{x,y}{x,y}$.
Periodic boundary conditions are imposed on the top and bottom boundaries, while Dirichlet and Neumann boundary conditions are applied to the left and right boundaries, respectively.
Under these conditions, $(D_x^+)^\dagger = -D_x^-$ and $(D_y^+)^\dagger = -D_y^-$ hold, which ensures that $A$ is anti-Hermitian.

In the numerical simulation, the spatial coordinates are represented using $8$ qubits, i.e., $n=4$ qubits per dimension.
Figure~\ref{fig:gaussian} (a) illustrates the Gaussian profile $c(x,y)$, and Figs.~\ref{fig:gaussian} (b) and (c) show its third-order ($K = K_x = K_y = 3$) and tenth-order Fourier series approximations, respectively.
As shown in Fig. \ref{fig:gaussian} (c), increasing the degree qualitatively improves the approximation accuracy; however, we choose a low degree, $K=3$, for the efficiency of the quantum circuit simulation.
We also note that, since the function considered here does not satisfy the periodic analyticity (the case (a) in Lemma \ref{lem:truncation_Fourier}), the approximation error is expected to exhibit only polynomial decay with respect to the degree.
The block-encoding for $C$ is constructed using the method introduced in Ref. \cite{rosenkranz2025Quantum}.
This construction requires $\mathcal{O}(d^D + D n \log d)$ two-qubit gates and $D \lceil \log_2(2 d + 1) \rceil$ ancilla qubits for the LCU representation of the Fourier basis functions, yielding $6$ ancilla qubits for the diagonal block-encoding in our case.
\begin{figure*}
    \centering
    \subfigure[]{
	\includegraphics[scale=.5]{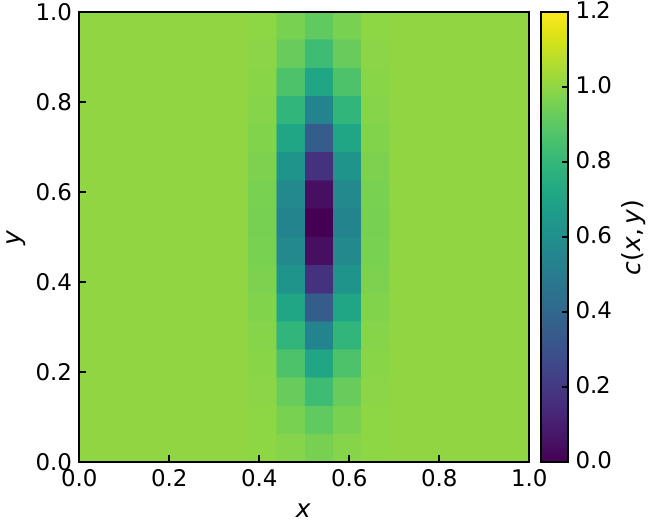}
    }
    \subfigure[]{
	\includegraphics[scale=.5]{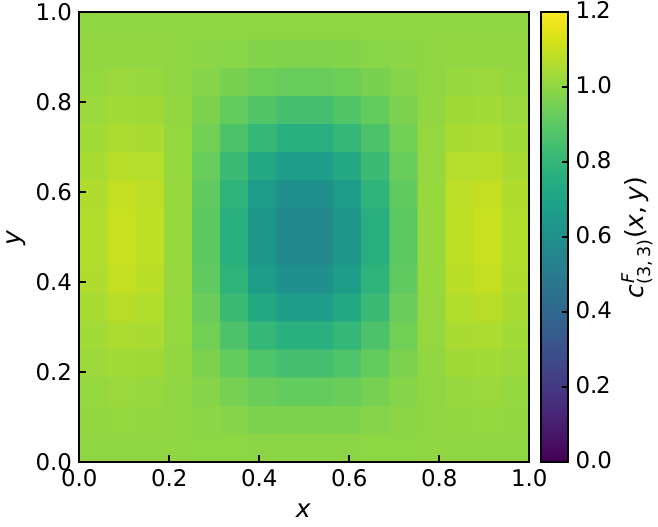}
    }
    \subfigure[]{
	\includegraphics[scale=.5]{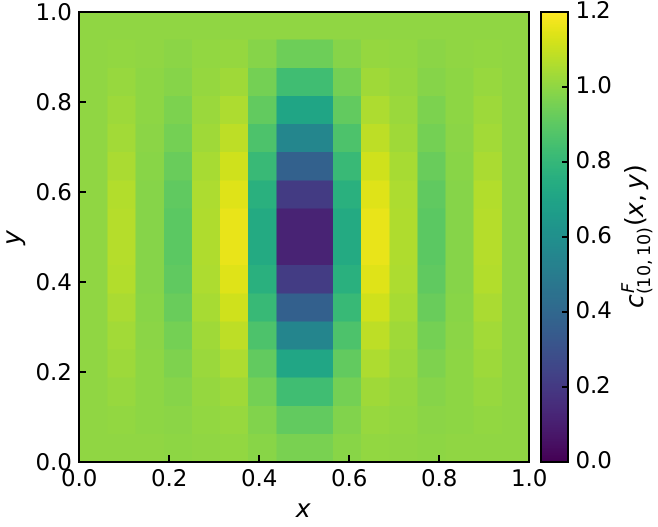}
    }
    \caption{The value of spatially varying coefficient $c(x,y)$ in the wave equation. (a) The original two-dimensional Gaussian profile $c(x,y)$, (b) its degree-three Fourier series approximation $c_{(3,3)}^F(x,y)$, and (c) its degree-$10$ Fourier series approximation $c_{(10,10)}^F(x,y)$. 
    The lower-order approximation (b) was used in the numerical simulation to reduce the scale of the quantum circuit simulation.}
    \label{fig:gaussian}
\end{figure*}
We set the initial conditions as 
\begin{align}
    c(x,y)^{-1} \frac{\partial u(0,x,y)}{\partial t} &= \begin{cases} 
      2^{-5/2} & x \geq (2^{4}-2)/(2^{4}-1) \\
      0 & x < (2^{4}-2)/(2^{4}-1)
   \end{cases}, \\
   \nabla_x u(0,x,y) &= 0 \\
   \nabla_y u(0,x,y) &= 0,
\end{align}
which ensures $\|u(0)\| = 1$, throughout this numerical demonstration.
Note that this initial state can be prepared using a trivial quantum circuit.

First, we conduct the forward simulation using the block-encoding of $C$.
We use four additional ancilla qubits to complete the block-encoding of $A$: one for the block-encoding of difference operators implemented via qml.BlockEncode (thereby avoiding an explicit construction to limit the simulation size), two for representing the system in four dimensions as in Eq. \eqref{eq:acoustic A}, and one for linearly combining $\ketbra{0}{1} \otimes C D_x^+ + \ketbra{1}{0} \otimes D_x^- C$ and $\ketbra{0}{2} \otimes C D_y^+ + \ketbra{2}{0} \otimes D_y^- C$.
Then, the block-encoding of $\exp(-At)$ can be constructed using QSVT with the use of the block-encoding of $A$ with two more ancilla qubits \cite[Corollary 62]{gilyen2019Quantum}, resulting in $20$ qubits in total.
From the result of the forward simulation at time $t = 1.0$ in Fig. \ref{fig:forward}, we observe that, in the wave propagating from right to left, the portion near $y=0.5$ travels more slowly, indicating that the Gaussian profile has been correctly encoded.
\begin{figure}
    \centering
    \includegraphics[scale=.6]{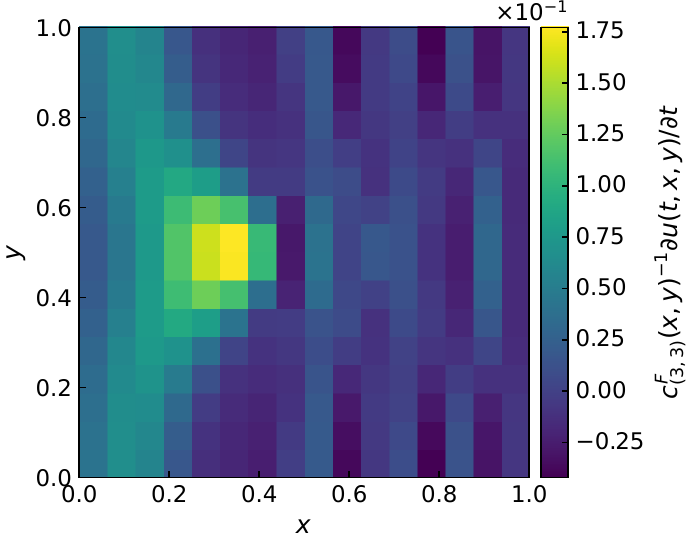}
    \caption{The value of $c_{(3,3)}^F(x,y)^{-1} \partial u(t,x,y)/\partial t$ at time $t = 1.0$.}
    \label{fig:forward}
\end{figure}

Then, we proceed to the parameter design problem by parameterizing the spatial shift of the two-dimensional Gaussian distribution as
\begin{equation}
    c(x,y;\xi) = - \exp(- \left(\frac{(x-\xi_x)^2}{2\cdot(1/20)^2} + \frac{(y-\xi_y)^2}{2\cdot(1/5)^2} \right)) + 1,
\end{equation}
for design parameters $\xi = (\xi_x, \xi_y)$.
We use $n_{\xi_x} = n_{\xi_y} = 4$ qubits, now resulting in $28$ qubits in total.
Note that we do not need any additional ancilla qubits for parameter design.
Here we want to maximize $c(x,y;\xi)^{-1} \partial u(t,x,y;\xi)/\partial t$ at time $t = 1.0$ at the target region depicted in Fig. \ref{fig:parameter design} (a) by adjusting the position of the Gaussian distribution.
That is, we consider the optimization problem
\begin{align}
    &\max_{\xi=(\xi_x,\xi_y)} \mathcal{F}(\xi), \quad \mathcal{F}(\xi) = \left| \sum_{(x,y) \in S} w_1(1,x,y;\xi)^2 \right|^{1/2},
\end{align}
where $w_1(t,x,y;\xi) = c(x,y;\xi)^{-1} \partial u(t,x,y;\xi)/\partial t$ and $S$ is a set of coordinate points of the target region.
The block-encoding of the objective function, $\hat{\mathcal{F}} = \sum_\xi \mathcal{F}(w(t;\xi)) \ketbra{\xi}{\xi}$, can be constructed using the same procedure discussed in Ref. \cite{sato2025Explicit}.
Fig. \ref{fig:parameter design}(b) and (c) show the value of the objective function evaluated using matrix calculation and quantum circuit simulation, respectively.
Although quantum circuit simulation has the error caused by the Hamiltonian simulation by QSVT, we observe that its landscape is qualitatively the same as the one with matrix calculation and the error is very small.
\begin{figure*}
  \begin{center}
  \subfigure[]{
	\includegraphics[scale=.5]{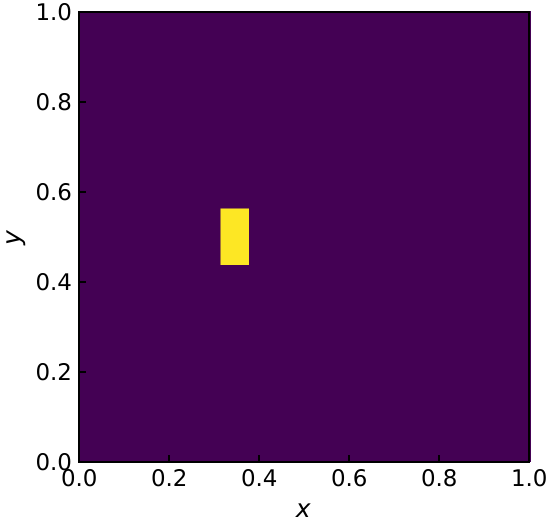}
  }
  \subfigure[]{
	\includegraphics[scale=.5]{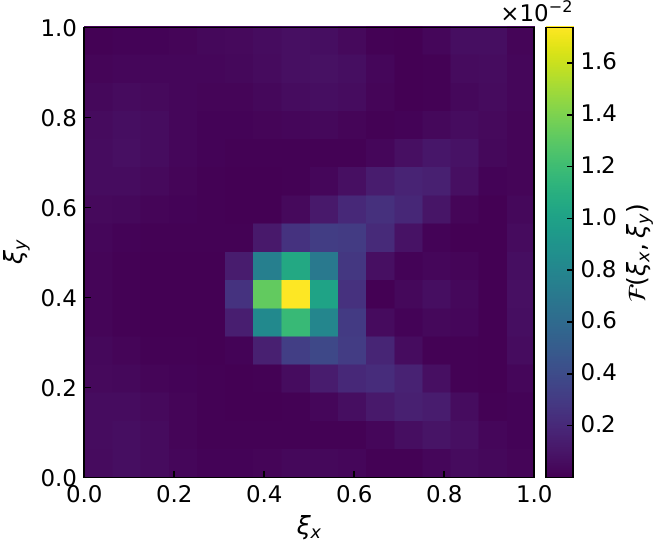}
  }
  \subfigure[]{
	\includegraphics[scale=.5]{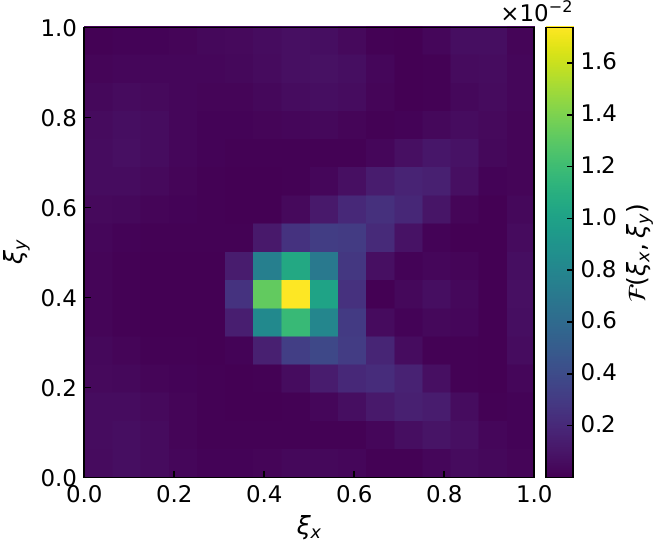}
  }
  \end{center}\vspace{-0.5cm}
  \caption{Numerical simulation results for the parameter design problem.
  (a) Target region (yellow). (b) Objective function $\mathcal{F}(\xi_x,\xi_y)$ evaluated by matrix computation. (c) Objective function $\mathcal{F}(\xi_x,\xi_y)$ evaluated by quantum circuit simulation.
  In (c), the block-encoded matrix yields $(2/\alpha_\mathrm{for}^2) \mathcal{F}(\xi_x,\xi_y) - 1$ according to \cite{sato2025Explicit}, where $\alpha_\mathrm{for}$ is the normalization constant for the forward simulation of the wave equation.
  For ease of comparison with (b), we plot the rescaled value $\mathcal{F}(\xi_x,\xi_y)$.}
  \label{fig:parameter design}
\end{figure*}

\section{Conclusion}\label{sec:conclusion}
In this work, we develop a quantum-algorithmic framework for parameterizing partial differential equations that leverages diagonal block-encodings to represent spatially varying coefficient fields.
This approach allows one to incorporate structured, potentially nontrivial profiles into quantum circuits without resorting to unstructured data loading methods like QRAM.
We begin by presenting a procedure for constructing block-encodings for operators associated with linear PDEs using diagonal block-encodings. 
We have shown that, in terms of gate complexity, PDE information can be encoded efficiently when spatially varying coefficients admit efficient diagonal block-encodings, in particular when the coefficients can be efficiently approximated by Fourier series.
Based on these constructions, we analyze possible parameterization strategies with respect to three standard data-loading oracles.
This perspective based on quantum oracles enables a wide range of parameterizations that are of practical relevance in engineering applications.
We then demonstrate through numerical simulations that the proposed framework successfully parameterizes PDEs.
Building on this framework, parameterized forward PDE simulation can be invoked coherently across parameter values in oracle-based quantum algorithms, thereby providing a pathway to higher-level tasks such as design and optimization \cite{sato2025Explicit}.

Several directions remain for future work.
First, while this paper focuses on parameterizing the time-evolution operator of linear PDEs, the framework can be extended to stationary problems and nonlinear PDEs, as well as to the optimization of initial conditions.
Second, beyond PDE-constrained optimization, it will be important to explore other potential uses of PDE-related oracles that leverage parameterized PDE simulations.

\begin{acknowledgments}
We thank Dr. Seiji Kajita for insightful comments.
\end{acknowledgments}

\appendix
\section{Proof of Theorems}
We start by recalling fundamental results on the product and linear combination of block-encodings.
\begin{lem}[Product of two block-encoded matrices, {\cite[Lemma 53]{gilyen2019Quantum}}]\label{lem:prod_of_BE}
    If $U_A$ is an $(\alpha_A,a_A,\epsilon_A)$-block-encoding of an $n$-qubit operator $A$, and $U_B$ is an $(\alpha_B,a_B,\epsilon_B)$-block-encoding of an $n$-qubit operator $B$, then $(I_b \otimes U_A)(I_a \otimes U_B)$ is an $(\alpha_A \alpha_B, a_A + a_B, \alpha_A \epsilon_B + \alpha_B \epsilon_A)$-block-encoding of $AB$.
\end{lem}
Owing to the above lemma, one can readily construct a $(\prod_{j=0}^{m-1} \alpha_j, \sum_{j=0}^{m-1} a_j, \sum_{j=0}^{m-1} \epsilon_j \prod_{k \neq j} \alpha_k)$-block-encoding of the multi-product $A_1 A_2 ... A_m$, given an $(\alpha_j, a_j, \epsilon_j)$-block-encoding of $A_j$, for $j = 0,...,m-1$.

We slightly modify \cite[Lemma 52]{gilyen2019Quantum} for the linear combination of block-encoded matrices.
\begin{lem}[Linear combination of block-encoded matrices]\label{lem:LC_of_BE}
    Let $y \in \mathbb{C}^m$.
    Let $U_j$ be an $(\alpha_j,a_j,\epsilon_j)$-block-encoding of an $n$-qubit operator $A_j$, for $j=0,...,m-1$.
    Let us denote $\alpha \coloneqq \max_j \alpha_j$, $a \coloneqq \max_j a_j$, $\epsilon \coloneqq \max_j \epsilon_j$
    Let us also denote $\alpha_\mathrm{prep} \geq \sum_j |y_j| (\alpha_j/\alpha)$ and $a_\mathrm{prep} \coloneqq \lceil\log_2(m)\rceil$.
    Suppose that $(P_L, P_R)$ is an $(\alpha_\mathrm{prep},a_\mathrm{prep},\epsilon_\mathrm{prep})$-state-preparation-pair for $z \in \mathbb{C}^m$ with $z_j \coloneqq y_j (\alpha_j/\alpha)$ where $P_L \ket{0} = \sum_{j=0}^{2^{a_\mathrm{prep}}-1} c_j \ket{j}$ and $P_R \ket{0} = \sum_{j=0}^{2^{a_\mathrm{prep}}-1} d_j \ket{j}$ such that $\sum_{j=0}^{m-1}|\alpha_\mathrm{prep}(c_j^* d_j) - z_j| \leq \epsilon_\mathrm{prep}$ and $c_j^* d_j = 0$ for $j = m,...,2^{a_\mathrm{prep}}-1$, and $W = \sum_{j=0}^{m-1} \ketbra{j}{j} \otimes (I^{\otimes a-a_j} \otimes U_j) + (I^{\otimes a_\mathrm{prep}} - \sum_{j=0}^{m-1} \ketbra{j}{j}) \otimes I^{\otimes a} \otimes I^{\otimes n}$.
    Then, we can implement an $(\alpha \alpha_\mathrm{prep}, a+a_\mathrm{prep}, \alpha \epsilon_\mathrm{prep} + \epsilon \|y\|_1)$-block-encoding of $\sum_{j=0}^{m-1} y_j A_j$.
\end{lem}
\begin{proof}
    Let $\tilde{W} = (P_L^\dagger \otimes I^{\otimes a+n}) W (P_R \otimes I^{\otimes a+n})$.
    Then, we have
    \begin{widetext}
    \begin{align}
        &\left\| \sum_j y_j A_j - \alpha \alpha_\mathrm{prep} (\bra{0}^{\otimes a_\mathrm{prep}} \otimes \bra{0}^{\otimes a} \otimes I^{\otimes n}) \tilde{W} (\ket{0}^{\otimes a_\mathrm{prep}} \otimes \ket{0}^{\otimes a} \otimes I^{\otimes n}) \right\| \\
        =& \left\| \sum_j y_j A_j - \alpha \sum_j \alpha_\mathrm{prep} (c_j^* d_j) (\bra{0}^{\otimes a} \otimes I^{\otimes n}) (I^{\otimes a-a_j} \otimes U_j) (\ket{0}^{\otimes a} \otimes I^{\otimes n}) \right\| \\
        \leq& \alpha \epsilon_\mathrm{prep} + \left\| \sum_j y_j A_j - \sum_j y_j \alpha_j (\bra{0}^{\otimes a} \otimes I^{\otimes n}) (I^{\otimes a-a_j} \otimes U_j) (\ket{0}^{\otimes a} \otimes I^{\otimes n}) \right\| \\
        \leq& \alpha \epsilon_\mathrm{prep} + \sum_j |y_j| \left\| A_j - \alpha_j (\bra{0}^{\otimes a} \otimes I^{\otimes n}) (I^{\otimes a-a_j} \otimes U_j) (\ket{0}^{\otimes a} \otimes I^{\otimes n}) \right\| \\
        \leq& \alpha \epsilon_\mathrm{prep} + \sum_j |y_j| \epsilon_j \\
        \leq& \alpha \epsilon_\mathrm{prep} + \epsilon \|y\|_1.
    \end{align}
    \end{widetext}
    This concludes the proof of the lemma.
\end{proof}
In the special case $\epsilon_\mathrm{prep} = 0$, the result takes a simpler form.
\begin{cor}\label{cor:LC_of_BE}
    When we choose a $(\|y\|_1,a_\mathrm{prep},0)$-state-preparation-pair, we can implement an $(\alpha \|y\|_1, a+a_\mathrm{prep}, \epsilon \|y\|_1)$-block-encoding of $\sum_{j=0}^{m-1} y_j A_j$.
\end{cor}
We use this corollary for the convenience of the algorithm-level analysis below, as usually done in most studies.
However, we note that, in the actual end-to-end FTQC resource/error accounting, we need to consider and allocate $\epsilon_\mathrm{prep}$ small enough that it does not affect the final error bound, i.e., $\epsilon_\mathrm{prep}$ is sufficiently small compared with the overall block-encoding error.

\subsection{Block-encoding for $A^{\mathrm{(2nd)}}$}\label{subsec:app_A2nd}
Prior to the proof of Theorem \ref{thm:query_A2nd}, we establish a lemma that exploits the structure of $A^{\mathrm{(2nd)}}$.
\begin{lem}\label{lem:BE_M_j}
    Assume $U_1$ and $U_2$ are $(\alpha,a,\epsilon)$-block-encodings of $A_1$ and $A_2$, respectively.
    Let $S$ be the selector register with $k$ qubits.
    Then, we can implement a $(\alpha,a,\epsilon)$-block-encoding of $M_j \coloneqq \ketbra{0}{j}_S \otimes A_1 + \ketbra{j}{0}_S \otimes A_2$ for any $j = 1,..., 2^k-1$.
\end{lem}
\begin{proof}
    We first add two work qubits $f_0, f_j$ initialized to $\ket{0}$.
    These qubits serve as flag qubits such that $f_j$ is $1$ if and only if the register $S$ equals $j$, for $j = 0,...,2^k-1$.
    This computation corresponds to AND operation and can be implemented using a ladder of $4k-2$ Toffolis \cite[Lemma 7.2]{barenco1995Elementary}.
    We first compute a flag, apply a single-controlled $U_1$ or $U_2$ conditioned on the flag, and then uncompute the flag.
    To specify the off-diagonal elements of $\ketbra{0}{j}_S$ (or $\ketbra{j}{0}_S$), a SWAP operation is also required on the selector register.
    The full quantum circuit implementation of the block-encoding of $M_j$ is shown in Fig \ref{fig:qc_BE_M_j}.
    \begin{figure*}
        \begin{center}
\begin{quantikz}[row sep=0.15cm, column sep=0.35cm, transparent]
\lstick{$\ket{0}_{f_j}$}
  & \qw
  & \qw
  & \qw
  & \gate[wires=3, label style={yshift=-0.65cm}]{\rotatebox{90}{$\mathrm{AND}_j$}}
  & \ctrl{3} %
  & \qw
  & \qw
  & \octrl{1} %
  & \qw
  & \gate[wires=3, label style={yshift=-0.65cm}]{\rotatebox{90}{$\mathrm{AND}_j^\dagger$}}
  & \qw
  & \qw \\

\lstick{$\ket{0}_{f_0}$}
  & \qw
  & \qw
  & \gate[wires=2]{\rotatebox{90}{$\mathrm{AND}_0$}}
  & \linethrough
  & \qw
  & \qw
  & \ctrl{2} %
  & \octrl{3}
  & \qw
  & \linethrough
  & \gate[wires=2]{\rotatebox{90}{$\mathrm{AND}_0^\dagger$}}
  & \qw \\

\lstick{$\ket{i}_S$}
  & \qwbundle{k}
  & \gate[wires=1]{\rotatebox{90}{$\mathrm{SWAP}_{0,j}$}}
  &
  &
  & \qw
  & \qw
  & \qw
  & \qw
  & \qw
  & 
  &
  & \qw \\

\lstick{$\ket{\psi}_{\mathrm{sys}}$}
  & \qwbundle{n}
  & \qw
  & \qw
  & \qw
  & \gate[wires=2]{U_1}
  & \qw
  & \gate[wires=2]{U_2}
  & \qw
  & \qw
  & \qw
  & \qw
  & \qw \\

\lstick{$\ket{0}_{\mathrm{anc}}$}
  & \qwbundle{a}
  & \qw
  & \qw
  & \qw
  & 
  & \qw
  & 
  & \gate{X \otimes I^{\otimes a-1}} %
  & \qw
  & \qw
  & \qw
  & \qw
\end{quantikz}
        \end{center}
        \caption{A quantum circuit for block-encoding of $M_j$. 
        Let $[k] = \{0, ..., k-1\}$ and an integer $j \in [2^k]$ represented by an $k$-bit string as $j = j_{k-1}...j_{0}$, i.e., $j = \sum_{i \in [k]} j_i 2^i$ with $j_i \in \{0,1\}$. We denote $\mathrm{AND}_j$ for an oracle $\ket{x}_S \ket{0}_{f_j} \mapsto \ket{x}_S \ket{f(x)}_{f_j}$ where $f(x) = \land_{i=1}^k (1 \oplus x_i \oplus j_i)$.
        $\mathrm{SWAP}_{0,j}$ is a unitary $U = \ketbra{0}{j} + \ketbra{j}{0} + \sum_{i \notin \{0,j\}} \ketbra{i}{i}$.
        The multi-qubit-controlled $X \otimes I^{\otimes a-1}$ is applied for encoding a zero to the top-left block on the irrelevant selector basis by using the fact that $\bra{0^a} X \otimes I^{\otimes a-1} \ket{0^a} = 0$.
        The inverse of $\mathrm{AND}_j$ and $\mathrm{AND}_0$ at the end of the circuit makes the flag qubits clean.}
        \label{fig:qc_BE_M_j}
    \end{figure*}
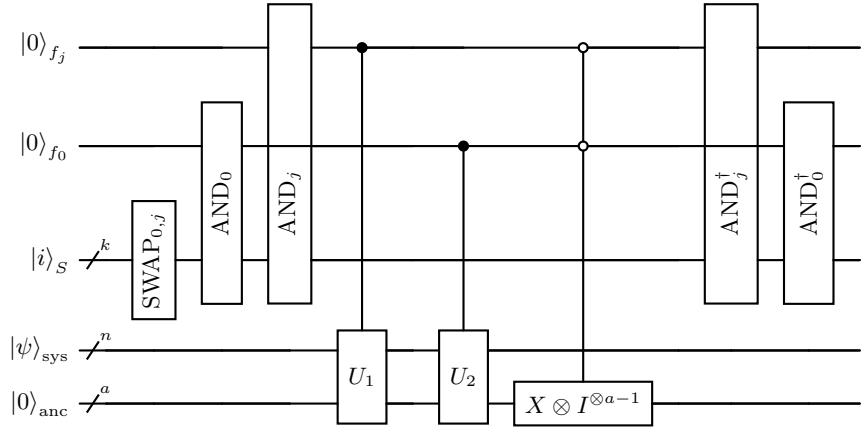
\end{proof}

Applying this lemma, we obtain the proof for the query complexity for $A^{\mathrm{(2nd)}}$.
\begin{proof}[Proof of Theorem \ref{thm:query_A2nd}]
    We first construct each term as a product of block-encoded matrices using Lemma \ref{lem:prod_of_BE}, resulting in $(\alpha_{\sqrt{\kappa}} \alpha_D \alpha_{\frac{1}{\sqrt{\varrho}}},a_{\sqrt{\kappa}}+a_D+a_{\frac{1}{\sqrt{\varrho}}},\alpha_D \alpha_{\sqrt{\kappa}} \epsilon_{\frac{1}{\sqrt{\varrho}}} + \alpha_{\frac{1}{\sqrt{\varrho}}} \alpha_D \epsilon_{\sqrt{\kappa}})$-block-encodings of $C_\varrho^{-\frac{1}{2}} D_\mu^+ C_\kappa^{\frac{1}{2}}$ and $C_\kappa^{\frac{1}{2}} D_\mu^- C_\varrho^{-\frac{1}{2}}$, an $(\alpha_{\frac{1}{\sqrt{\varrho}}} \alpha_{\sqrt{\gamma}},a_{\frac{1}{\sqrt{\varrho}}}+a_{\sqrt{\gamma}},\alpha_{\frac{1}{\sqrt{\varrho}}} \epsilon_{\sqrt{\gamma}} + \alpha_{\sqrt{\gamma}} \epsilon_{\frac{1}{\sqrt{\varrho}}})$-block-encoding of $C_\varrho^{-\frac{1}{2}} C_\gamma^{\frac{1}{2}}$.
    Squaring $C_\varrho^{-\frac{1}{2}}$ by QSVT with a single use of $U_{\frac{1}{\sqrt{\varrho}}}$ and its inverse, we have an $(\alpha_{\frac{1}{\sqrt{\varrho}}}^2,a_{\frac{1}{\sqrt{\varrho}}}+1, 8\sqrt{\epsilon_{\frac{1}{\sqrt{\varrho}}}})$-block-encoding of $C_\varrho^{-1}$, based on \cite[Lemma 22]{gilyen2019Quantum}.
    Thus, we can construct an $(\alpha_{\frac{1}{\sqrt{\varrho}}}^2 \alpha_\zeta,a_{\frac{1}{\sqrt{\varrho}}}+a_\zeta+1,\alpha_{\frac{1}{\sqrt{\varrho}}}^2 \epsilon_\zeta + 8 \alpha_\zeta \sqrt{\epsilon_{\frac{1}{\sqrt{\varrho}}}})$-block-encoding of $C_\varrho^{-1} C_\zeta$.
    Then, based on Lemma \ref{lem:BE_M_j} and Corollary \ref{cor:LC_of_BE}, we add a selector register and linearly combine all the terms to build a $(\alpha \|y\|_1,a+a_\mathrm{prep},\epsilon \|y\|_1)$-block-encoding of $A^\mathrm{(2nd)}$ with $\alpha = \max(\alpha_{\frac{1}{\sqrt{\varrho}}}^2 \alpha_\zeta, \alpha_{\sqrt{\kappa}} \alpha_D \alpha_{\frac{1}{\sqrt{\varrho}}}, \alpha_{\frac{1}{\sqrt{\varrho}}} \alpha_{\sqrt{\gamma}})$, $a = \max(a_{\sqrt{\kappa}}+a_D+a_{\frac{1}{\sqrt{\varrho}}}, a_{\frac{1}{\sqrt{\varrho}}}+a_{\sqrt{\gamma}}, a_{\frac{1}{\sqrt{\varrho}}}+a_\zeta+1)$, $a_\mathrm{prep}=\lceil \log_2(d+2) \rceil$, $\epsilon = \max(\alpha_D \alpha_{\sqrt{\kappa}} \epsilon_{\frac{1}{\sqrt{\varrho}}} + \alpha_{\frac{1}{\sqrt{\varrho}}} \alpha_D \epsilon_{\sqrt{\kappa}}, \alpha_{\frac{1}{\sqrt{\varrho}}} \epsilon_{\sqrt{\gamma}} + \alpha_{\sqrt{\gamma}} \epsilon_{\frac{1}{\sqrt{\varrho}}}, \alpha_{\frac{1}{\sqrt{\varrho}}}^2 \epsilon_\zeta + 8 \alpha_\zeta \sqrt{\epsilon_{\frac{1}{\sqrt{\varrho}}}})$, and $\|y\|_1 = d+2$.
    
\end{proof}
By combining the gate-complexity results for each block-encoding, we derive the gate complexity of $A^{\mathrm{(2nd)}}$.
\begin{proof}[Proof of Theorem \ref{thm:gate_A2nd}]
    According to Ref. \cite{rosenkranz2025Quantum}, the LCU Fourier requires $\mathcal{O}(K^d + n \log K)$ two-qubit gates for the block-encoding of an $n$-qubit operator $C_f$ with a maximal degree-$K$ Fourier approximation of a $d$-dimensional function $f$.
    Ref. \cite[Appendix A]{sato2025Explicit} shows that the block-encoding of finite difference operators requires $\mathcal{O}(n^2/d)$ two-qubit gates.
    Constructing the block-encoding of the form $M_j$ in Lemma \ref{lem:BE_M_j}, together with the SELECT operation in the outer LCU, adds only $\mathcal{O}(\log d)$ two-qubit gates. (Specifically, one may compute the AND of the $\lceil \log_2(d+2) \rceil$ controls into a single flag qubit, apply a singly-controlled block-encoding, and then uncompute the flag.)
    The PREP operation in the outer LCU requires $\mathcal{O}(d)$ two-qubit gates.
    Combining these costs with the query complexity in Theorem \ref{thm:query_A2nd}, the total two-qubit gate complexity is $\mathcal{O}(d K^d + d n \log K + n^2)$.
\end{proof}

\subsection{Block-encoding for $A^{\mathrm{(1st)}}$}\label{subsec:app_A1st}
The query complexity of the block-encoding of $A^{\mathrm{(1st)}}$ can be derived by simply constructing each term as a product and subsequently forming a linear combination of all terms as follows:
\begin{proof}[Proof of Theorem \ref{thm:query_A1st}]
    First, by Lemma \ref{lem:prod_of_BE}, we construct $(\alpha_\kappa \alpha_D^2, a_\kappa+2a_D,\alpha_D^2 \epsilon_\kappa)$-block-encodings of $D_\mu^+ C_\kappa D_\mu^-$ and $D_\mu^- C_\kappa D_\mu^+$, as well as $(\alpha_\beta \alpha_D, a_\beta+a_D,\alpha_D \epsilon_\beta)$-block-encodings of $C_{\beta_\mu^+} D_\mu^-$ and $C_{\beta_\mu^-} D_\mu^+$.
    Next, applying Corollary \ref{cor:LC_of_BE}, we form a linear combination of block-encoded matrices over all the $4d+1$ terms, to obtain a $(\alpha \|y\|_1,a+a_\mathrm{prep},\epsilon \|y\|_1)$-block-encoding of $A^\mathrm{(1st)}$ with $\alpha = \max(\alpha_\kappa \alpha_D^2,\alpha_\beta \alpha_D,\alpha_\gamma)$, $a = \max(a_\kappa+2a_D,a_\beta+a_D,a_\gamma)$, $\epsilon = \max(\alpha_D^2 \epsilon_\kappa, \alpha_D \epsilon_\beta, \epsilon_\gamma)$, $a_\mathrm{prep} = \lceil \log_2(4d+1) \rceil$, and $\|y\|_1 = 3d+1$.
\end{proof}

Clarifying the gate complexity of each block-encoding readily yields the total gate complexity.
\begin{proof}[Proof of Theorem \ref{thm:gate_A1st}]
    The same line of arguments as in the case of the proof of Theorem \ref{thm:gate_A2nd} leads to a corresponding result.
\end{proof}

The truncation degree $K$ can be represented in terms of the approximation error $\epsilon_f$ when we make some assumption on the target function $f$, as in the following lemma:
\begin{lem}[Value of $K$ for Fourier series]\label{lem:truncation_Fourier}
    Assume $f$ is Lipschitz continuous periodic function on $[0, 2 \pi]$.
    For Fourier series approximation of $f$, to bound the approximation error by $\epsilon_f$ (in the infinity norm), the degree of Fourier series needs to be 
    \begin{enumerate}[label=(\alph*)]
    \item $\mathcal{O}(\log(1/\epsilon_f))$ when $f$ is analytic with $f|(\cdot)| \leq M$ in the open strip of half-width $\alpha$ around the real axis in the complex $t$-plane, or 
    \item $\mathcal{O}((1/\epsilon_f)^{1/\nu})$ when $f$ is $\nu \geq 1$ times differentiable and $f^{(\nu)}$ is of bounded variation $V$ on $[0,2 \pi]$.
    \end{enumerate}
\end{lem}
\begin{proof}
    See \cite[Theorem 4.2]{wright2015Extension} for the proof.
\end{proof}

Using this lemma, the gate complexity can be further detailed to bound the block-encoding error of $A^{\mathrm{(1st)}}$.
\begin{cor}\label{cor:gate_A1st_example}
    Consider case (a) and (b) in Lemma \ref{lem:truncation_Fourier}.
    Assume $\max(\alpha_D^2 \epsilon_\kappa, \alpha_D \epsilon_\beta, \epsilon_\gamma) = \alpha_D^2 \epsilon_\kappa$.
    To bound the block-encoding error of $A^{(1st)}$ by $\epsilon$, i.e., $(3d+1) \alpha_D^2 \epsilon_\kappa \leq \epsilon$, we need 
    \begin{itemize}
        \item $\mathcal{O}(d (\log d + \frac{n}{d} + \log(1/\epsilon))^d + d n (\log d + \frac{n}{d} + \log(1/\epsilon)) + n^2)$ two-qubit gates when $\kappa(\cdot)$ satisfies (a), and
        \item $\mathcal{O}(d^{3 \frac{d}{\nu} + 1} 4^{\frac{n}{\nu}} (1/\epsilon)^{\frac{d}{\nu}} + d n (1/\nu) \log(1/\epsilon) + n^2)$ two-qubit gates when $\kappa(\cdot)$ satisfies (b).
    \end{itemize}
\end{cor}
\begin{proof}
    According to Ref. \cite[Appendix A]{sato2025Explicit}, $\alpha_D = 3 d / h$.
    Then, to bound the block-encoding error of $A^{\mathrm{(1st)}}$ by $\epsilon$, $\epsilon_\kappa = \mathcal{O}((h^2/d^3) \epsilon)$.
    Inserting this into the order of truncation in Lemma \ref{lem:truncation_Fourier} and using it to the result of Theorem \ref{thm:gate_A1st} with $h = \mathcal{O}(2^{-n/d})$ completes the proof.
\end{proof}

\bibliography{main}

\end{document}